\documentclass[11pt]{article}

\usepackage[letterpaper,margin=1in]{geometry}
\usepackage{amsmath,amssymb,amsthm}
\usepackage{microtype}
\usepackage{hyperref}
\usepackage{xcolor}

\hypersetup{
  colorlinks=true,
  linkcolor=blue!55!black,
  citecolor=blue!55!black,
  urlcolor=blue!55!black,
  pdftitle={Parity Tests under Ties: A One-Test Lifting Theorem},
  pdfauthor={Ron Kupfer}
}

\theoremstyle{plain}
\newtheorem{theorem}{Theorem}[section]
\newtheorem{lemma}[theorem]{Lemma}
\newtheorem{corollary}[theorem]{Corollary}

\newcommand{\R}{\mathbb{R}}
\newcommand{\sgn}{\operatorname{sgn}}

\title{Parity Tests under Ties: A One-Test Lifting Theorem}
\author{Ron Kupfer}
\date{}

\begin{document}
\maketitle

\begin{abstract}
In the unrestricted polynomial decision-tree model, only the number of
polynomial sign tests is charged. A parity test asks for the sign of a product
of pairwise differences. Such tests
underlie low-depth randomized algorithms for maximum finding and top-\(k\)
selection, but their usual analysis assumes distinct inputs because a tie makes
the product vanish. We give a black-box lifting theorem that removes this
assumption. After \(O(\log n)\) polynomial tests determine the number of
nonzero pairwise differences, every subsequent parity test is simulated by one
polynomial test, consistently with a fixed lexicographic tie-breaking order.
The simulator is an elementary symmetric polynomial in masked first and second
powers of all pairwise differences. Thus a depth-\(D\) parity-test tree on
distinct inputs becomes a polynomial decision tree of depth \(D+O(\log n)\)
on arbitrary inputs, with no increase in randomized pointwise error for
order-selection problems. We obtain maximum finding in depth
\(O(\log n[\log n+\log(1/\delta)])\) with error \(\delta\), and top-\(k\)
selection in depth \(O(\log^2 n+k\log n)\) with inverse-polynomial error, both
without any promise on ties.
\end{abstract}

\section{Introduction}

In the unrestricted polynomial decision-tree model, a query returns
\(\sgn p(x)\in\{-1,0,+1\}\) for an arbitrary real polynomial \(p\); only query
depth is charged. Rabin's complete-proof method gives the sharp \(n-1\) deterministic
depth lower bound for maximum finding, with the relevant generic-width
argument formalized by Monta\~na, Pardo, and Recio
\cite{Rabin1972,MPR94}. Randomization behaves very differently. Ting and
Yao~\cite{TY94} find the maximum of \(n\) distinct reals with
\(O((\log n)^2)\) polynomial tests and inverse-polynomial error. Lam and
Ting~\cite{LT98} find the \(k\) largest distinct inputs with
\(O(\log^2 n+k\log n)\) tests. Their primitive is a \emph{parity test} of the
form
\[
  \sgn\prod_{(a,b)\in C}(x_a-x_b),
\]
where \(C\) is a collection of ordered pairs; the star case
\(C=\{(s,j):j\in B\}\) is central to maximum finding. Ting's earlier thesis
also studies selection under parity-like tests~\cite{TingThesis}.

A tie makes a literal parity query vanish.
Naor and Ruah~\cite[Secs.~4.3.1 and 6]{NaorRuah2001} report that an
unpublished 1996 manuscript by Ben-Or gives an \(O((\log n)^2)\)-test
maximum-finding algorithm for inputs with ties. Their brief description
counts equalities within sampled subsets, but its printed
squared-difference expression cannot itself determine parity, and the
description does not establish the full query bound. We have not examined
the manuscript and do not claim the tied-input \(O((\log n)^2)\) maximum
bound as new. Separately, Ben-Or supplied the complete
\(O((\log n)^3)\)-depth construction reported in our preliminary
note~\cite{KupferTies}; it spends \(O(\log n)\) tests per parity query.
Here we give a black-box simulation after \(O(\log n)\) shared preprocessing,
with one test per parity query and total depth \(D+O(\log n)\) for any
depth-\(D\) parity-test tree.

Specifically, we canonically orient all
\(N=\binom n2\) pairwise differences and let \(K\) be the number that are
nonzero. The value \(K\) is found by binary search in \(O(\log N)=O(\log n)\)
tests; no tie-class sizes are assumed known. Assign one masked value per pair:
its difference if the query uses that pair an odd number of times, and its
square otherwise, including for absent pairs.

Squaring preserves the zero set. Upon evaluation at the input, the \(K\)th
elementary symmetric
polynomial in these \(N\) masked differences
has exactly one nonzero summand: the product over all nonzero masked
differences.
Its sign is the requested parity when every tie is resolved by index. This one
identity yields a black-box compiler, not merely a maximum-finding gadget.
To our knowledge, this global one-test realization and the resulting additive
\(O(\log n)\) lifting theorem have not appeared previously.

\begin{theorem}[Lifting theorem, informal]
\label{thm:informal}
Every depth-\(D\) randomized parity-test decision tree for distinct real inputs
has a polynomial decision-tree simulation on arbitrary real inputs of depth
\[
  D+\left\lceil\log_2\!\left(\binom n2+1\right)\right\rceil.
\]
On a tied input, the simulated transcript is exactly the original transcript
under the fixed rule that the smaller index wins every tie.
\end{theorem}

The formal statement, including the condition under which output correctness
transfers, is Theorem~\ref{thm:lifting}. Maximum and top-\(k\) selection satisfy
that condition: a solution for a strict refinement of the input order remains
valid before the ties are broken. The construction uses polynomials of degree
at most \(2\binom n2\); degree and representation size are uncharged in the
model. All these polynomials nevertheless admit standard polynomial-size
arithmetic circuits; the elementary-symmetric dynamic program uses
\(O(NK)\) gates. This does not imply an analogous bound for bounded-degree
tests or for models charging arithmetic operations.

\section{The one-test simulator}
\label{sec:simulator}

Let
\[
  \mathcal E=\{\{u,v\}:1\le u<v\le n\},\qquad N=|\mathcal E|=\binom n2,
\]
and orient each edge \(e=\{u,v\}\), \(u<v\), by
\(d_e=x_u-x_v\). Define a strict total order \(\succ_x\) on the indices by
\begin{equation}
  \label{eq:tie-order}
  u\succ_x v
  \quad\Longleftrightarrow\quad
  x_u>x_v\ \text{ or }\ (x_u=x_v\text{ and }u<v).
\end{equation}
For an oriented difference \(x_a-x_b\), define its \emph{virtual sign} to be
\(+1\) if \(a\succ_x b\) and \(-1\) otherwise. For a product
\(c\prod_\ell(x_{a_\ell}-x_{b_\ell})\) with \(c\ne0\), its virtual sign is
\(\sgn(c)\) times the product of its factors' virtual signs. In particular,
the canonical difference \(d_e\) has virtual sign \(\sgn d_e\) when
\(d_e\ne0\), and \(+1\) when \(d_e=0\).

For any \(L\ge0\) and variables \(z_1,\ldots,z_L\), write
\[
  e_t(z_1,\ldots,z_L)
  =\sum_{\substack{S\subseteq[L]\\|S|=t}}\prod_{j\in S}z_j,
  \qquad e_0=1.
\]

\begin{lemma}[Global preprocessing]
\label{lem:preprocess}
Let \(K=|\{e\in\mathcal E:d_e\ne0\}|\). The value \(K\) is determined by
\(\lceil\log_2(N+1)\rceil\) polynomial tests.
\end{lemma}

\begin{proof}
For \(0\le t\le N\), put
\begin{equation}
  H_t=e_t\bigl((d_e^2)_{e\in\mathcal E}\bigr).
  \label{eq:threshold}
\end{equation}
Every monomial is nonnegative, and a positive monomial exists exactly when at
least \(t\) differences are nonzero. Hence \(H_t>0\) if and only if
\(t\le K\). Binary search for this threshold among \(N+1\) possibilities
proves the claim.
\end{proof}

A parity query is the sign of
\begin{equation}
  \label{eq:parity-query}
  P=c\prod_{\ell=1}^r(x_{a_\ell}-x_{b_\ell}),
  \qquad r\ge0,\quad c\in\R\setminus\{0\},\quad a_\ell\ne b_\ell,
\end{equation}
where repetitions and either orientation are allowed; when \(r=0\), the
product is \(1\). Let \(S\subseteq\mathcal E\) contain precisely the
undirected edges occurring an odd number of times. Set \(\sigma\) to
\(\sgn(c)\), with its sign flipped for every reversed factor
\((x_a-x_b)\) with \(a>b\). Since each comparison sign under \(\succ_x\)
is \(\pm1\), repeated factors cancel in pairs. On distinct inputs the
query has the same sign as the squarefree query
\(P^\flat=\sigma\prod_{e\in S}d_e\). Their virtual signs under
\(\succ_x\) also agree. At a tie, this virtual sign is different from the
literal sign of a polynomial that vanishes. For the simulator set
\[
  y_e=\begin{cases}
    d_e,&e\in S,\\
    d_e^2,&e\notin S.
  \end{cases}
\]

\begin{lemma}[One-test simulation]
\label{lem:one-test}
Once \(K\) is known, the sign of the parity query
\eqref{eq:parity-query} under the total order \(\succ_x\) is obtained by the
single polynomial test
\begin{equation}
  Q_P=\sigma\,e_K\bigl((y_e)_{e\in\mathcal E}\bigr).
  \label{eq:simulator}
\end{equation}
Moreover, \(Q_P(x)\ne0\).
\end{lemma}

\begin{proof}
There is one \(y_e\) per edge, whether or not the edge occurs in the query.
Because \(d_e^2=0\) exactly when \(d_e=0\), we have
\(y_e\ne0\) if and only if \(d_e\ne0\). Thus exactly \(K\) values
\(y_e\) are nonzero. The unique nonzero
summand of \(e_K\) in \eqref{eq:simulator} selects every edge with
\(d_e\ne0\), and
\[
  Q_P(x)=\sigma\prod_{e:d_e\ne0}y_e\ne0.
\]
Its sign is
\[
  \sigma\prod_{e\in S:d_e\ne0}\sgn d_e.
\]
This is the virtual sign of \(P^\flat\), hence of
\eqref{eq:parity-query}: every tied canonical comparison contributes \(+1\)
as in \eqref{eq:tie-order}.
\end{proof}

For instance, if \(x=(1,1,0)\) and
\(P=(x_2-x_1)(x_1-x_3)\), then \(K=2\), \(\sigma=-1\), and
\(Q_P=-e_2(0,1,1)=-1\), as prescribed by \(1\succ_x2\succ_x3\).

\paragraph{Fixed factor families.}
The same masking identity works for any fixed family of polynomials
\(f_1,\ldots,f_M\) with prescribed signs \(\alpha_j\in\{-1,+1\}\) at zeros:
\(\lceil\log_2(M+1)\rceil\) tests count the nonzero factors, after which one
test evaluates the prescribed sign of any product of these factors. Such
assignments need not describe one perturbed input: at a tie, assigning
\(+1\) to both \(x_a-x_b\) and \(x_b-x_a\) is inconsistent with every strict
order. The tie-breaking signs in \eqref{eq:tie-order} are coherent, which
enables the lifting theorem.

For a single star query, the earlier local simulator~\cite{KupferTies} has a
complementary input-sensitive bound. Let \(b=|B|\) and
\(\mu_B=|\{a\in B:x_a=x_i\}|\). For \(1\le t\le b\), the polynomial
\(e_{b-t+1}(((x_i-x_a)^2)_{a\in B})\) vanishes exactly when
\(t\le\mu_B\). Exponential search followed by binary search finds
\(\mu_B\) in \(O(1+\log(\mu_B+1))\) tests; one test of
\(e_{b-\mu_B}((x_i-x_a)_{a\in B})\) then returns the parity with tied
factors deleted. This local bound can be smaller for an isolated query; the
global preprocessing is shared over the entire tree.

\section{A black-box lifting theorem}
\label{sec:lifting}

A parity-test decision tree is a binary decision tree whose internal queries
have the form \eqref{eq:parity-query}; on distinct inputs, no query is zero.
We allow randomization
by taking an input-independent distribution over deterministic trees. The
pointwise error is computed for a fixed input over this distribution.

Let \(R(x,o)\) specify whether output \(o\) is valid on input \(x\). We say
\(R\) is \emph{stable under tie refinement} if, whenever \(z\) has distinct
coordinates and \(x_i>x_j\) implies \(z_i>z_j\), then
\(R(z,o)\) implies \(R(x,o)\) for every output \(o\). Maximum finding and
selection of the \(k\) largest indices have this property; returning the
entire set of maximizers does not.

\begin{theorem}[Black-box lifting]
\label{thm:lifting}
Let \(\mathcal A\) be a possibly randomized parity-test decision tree of depth
\(D\) on distinct inputs. There is a polynomial decision tree
\(\widetilde{\mathcal A}\) of depth
\[
  D+\left\lceil\log_2(N+1)\right\rceil,
\]
such that, for every \(x\in\R^n\), its output distribution equals the output
distribution of \(\mathcal A\) on any distinct input whose index order is
\(\succ_x\). Consequently, pointwise error does not increase for any search
relation stable under tie refinement: the error on \(x\) is at most the error
of \(\mathcal A\) on any such distinct input.
\end{theorem}

\begin{proof}
Prepend the binary search of Lemma~\ref{lem:preprocess}. For each deterministic
tree in the support of \(\mathcal A\), replace every query
\eqref{eq:parity-query} by the polynomial
\eqref{eq:simulator}. At each new node, \(K\) and the original query are fixed
by the preceding transcript, so this is a legal adaptive polynomial test.
Lemma~\ref{lem:one-test} shows that its zero branch is never reached and that
its sign is the original query sign under \(\succ_x\).

For all sufficiently small \(\epsilon>0\), the distinct vector
\(z_i=x_i-\epsilon i\) has index order \(\succ_x\).
This single \(\epsilon\) depends only on \(x\), not on the adaptive queries
or random choices; the simulator does not need to compute it.
Induction down the tree shows that the lifted transcript on \(x\) equals the
original transcript on \(z\). On distinct inputs, every parity-query sign, and
hence the entire transcript, depends only on the strict index order. Thus the
same conclusion holds for any distinct vector with order \(\succ_x\). Apply
this construction separately to every
deterministic tree in the input-independent distribution for \(\mathcal A\).
The output distributions agree, and \(R(z,o)\) implies \(R(x,o)\) by
tie-refinement stability. Thus the error probability on \(x\) is no larger
than that on \(z\).
\end{proof}

The theorem concerns products of pairwise differences, not arbitrary
polynomial queries. It is a transcript simulation; correctness need not
transfer for problems that ask about the ties themselves, such as recovering
all equality classes.

\section{Applications}
\label{sec:applications}

\subsection{Maximum finding with arbitrary error}

Theorem~\ref{thm:lifting} preserves pointwise error. The following analysis of
Ting and Yao's pivot-ascent strategy~\cite{TY94} gives an explicit bound for
any target \(\delta\). Put
\[
  m=\lceil\log_2(n+1)\rceil,
  \qquad q=\lceil\log_2(1/\delta)\rceil,
  \qquad T=16(m+q).
\]
Run the preprocessing of Lemma~\ref{lem:preprocess} once. Choose an arbitrary
initial pivot \(i\). Call a positive parity sign even and a negative sign odd.
In each of \(T\) trials, include every \(a\ne i\) independently in a set
\(B\) with probability \(1/2\), and simulate the parity query
\(\prod_{a\in B}(x_i-x_a)\).

An even parity retains \(i\). For an odd parity, give \(B\) a fresh uniform
random order and repeatedly bisect it into two nearly equal parts. Query one
part and retain the odd part; the other parity is inferred because the current
parity is odd. After at most \(m\) such queries, a singleton remains. Make it
the new pivot. After \(T\) trials, output the pivot.

All parities in this procedure are understood under \(\succ_x\) and are
implemented by Lemma~\ref{lem:one-test}. If
\(U_i=\{a:a\succ_x i\}\), an odd trial localizes an element of \(U_i\).
When \(U_i\ne\varnothing\), fair sampling gives
\(\Pr[|B\cap U_i|\text{ is odd}]=1/2\). Conditional on this event and on the
sampled set \(B\), the fresh uniform ordering makes localization symmetric over
\(B\cap U_i\). Averaging over the exchangeable fair sample makes the returned
element uniform on \(U_i\). At least half of \(U_i\) has at most \(|U_i|/2\)
elements above it. Thus every trial at a nonmaximum pivot is a rank-halving
ascent with conditional probability at least \(1/4\).

Let \(\mathcal F_{t-1}\) be the history before trial \(t\). Set \(Z_t=1\) if
the pivot entering trial \(t\) is maximal, and otherwise let \(Z_t\) indicate
that the trial is a halving ascent. Then
\(\Pr[Z_t=1\mid\mathcal F_{t-1}]\ge1/4\). The usual sequential coupling with
fresh uniform thresholds couples \(Z_t\) from below by independent
\(\operatorname{Bernoulli}(1/4)\) variables. Hence \(\sum_tZ_t\) stochastically
dominates \(X\sim\operatorname{Bin}(T,1/4)\). Every update is an ascent, and
\(m\) halving ascents force \(|U_i|=0\), since initially \(|U_i|<2^m\).
Hence error implies \(\sum_tZ_t<m\). As
\(\mathbb EX=4(m+q)\), Chernoff's inequality gives
\[
  \Pr[\text{error}]
  \le\Pr[X<m]
  \le\exp(-9(m+q)/8)
  \le2^{-q}
  \le\delta.
\]
The preprocessing costs at most \(2m\) tests, and each trial costs at most
\(m+1\). Hence the worst-case depth is
\[
  2m+16(m+q)(m+1)\le25m(m+q).
\]

\begin{corollary}[Maximum finding]
\label{cor:max}
For every \(n\ge2\) and \(0<\delta\le1/4\), maximum finding on arbitrary real
inputs has a randomized polynomial decision tree of depth
\[
  O\bigl(\log n\,[\log n+\log(1/\delta)]\bigr)
\]
and pointwise error at most \(\delta\).
\end{corollary}

The maximum under \(\succ_x\) is the smallest-index member of the actual
maximum class, so the reported index is valid for the original input. For
inverse-polynomial error, the depth is \(O((\log n)^2)\), matching the known
general-position upper bound~\cite{TY94}. No matching randomized
\(\Omega((\log n)^2)\) depth lower bound is known in the unrestricted-degree
model.

\subsection{Top-\texorpdfstring{\(k\)}{k} selection}

Call \(S\subseteq[n]\), \(|S|=k\), a valid top-\(k\) output when
\(x_i\ge x_j\) for every \(i\in S\) and \(j\notin S\). Lam and Ting's
randomized algorithm on distinct inputs uses
\(O(\log^2 n+k\log n)\) parity tests and has error \(O(n^{-c})\) for every
fixed \(c>1\)~\cite{LT98}. Its queries are precisely products of pairwise
differences, so Theorem~\ref{thm:lifting} applies. The \(k\) largest indices
under \(\succ_x\) form a valid top-\(k\) output for \(x\), including when a
tie straddles the boundary.

\begin{corollary}[Top-\(k\) selection]
\label{cor:topk}
For every fixed \(c>1\) and \(1\le k\le n\), a valid top-\(k\) set of an
arbitrary real input can be selected with pointwise error \(O(n^{-c})\) by a
randomized polynomial decision tree of depth
\[
  O(\log^2 n+k\log n).
\]
\end{corollary}

\section*{Acknowledgments}

The author thanks Michael Ben-Or for sharing the complete earlier
\(O((\log n)^3)\)-depth construction reported in the preliminary
note~\cite{KupferTies}, which motivated this work.

\section*{AI use disclosure}

ChatGPT was used to explore the global lifting construction, check the
algebraic and probabilistic arguments, audit related work, and edit the
exposition. The author takes responsibility for the results and text.


\begingroup
\small

\begin{thebibliography}{9}
\footnotesize
\setlength{\itemsep}{1pt}

\bibitem{Rabin1972}
Michael O. Rabin.
\newblock Proving simultaneous positivity of linear forms.
\newblock \emph{Journal of Computer and System Sciences}, 6(6):639--650,
  1972.
\newblock \url{https://doi.org/10.1016/S0022-0000(72)80034-5}.

\bibitem{MPR94}
Jos\'e L. Monta\~na, Luis M. Pardo, and Tom\'as Recio.
\newblock A note on Rabin's width of a complete proof.
\newblock \emph{Computational Complexity}, 4(1):12--36, 1994.
\newblock \url{https://doi.org/10.1007/BF01205053}.

\bibitem{TY94}
Hing F. Ting and Andrew C. Yao.
\newblock A randomized algorithm for finding maximum with
  \(O((\log n)^2)\) polynomial tests.
\newblock \emph{Information Processing Letters}, 49(1):39--43, 1994.
\newblock \url{https://doi.org/10.1016/0020-0190(94)90052-3}.

\bibitem{LT98}
Tak Wah Lam and Hing Fung Ting.
\newblock Selecting the \(k\) largest elements with parity tests.
\newblock In \emph{Algorithms and Computation (ISAAC 1998)}, volume 1533 of
  \emph{Lecture Notes in Computer Science}, pages 189--197. Springer, 1998.
\newblock \url{https://doi.org/10.1007/3-540-49381-6_21}.

\bibitem{NaorRuah2001}
Moni Naor and Sitvanit Ruah.
\newblock On the decisional complexity of problems over the reals.
\newblock \emph{Information and Computation}, 167(1):27--45, 2001.
\newblock \url{https://doi.org/10.1006/inco.2000.3012}.

\bibitem{TingThesis}
Hing Fung Ting.
\newblock Computational complexity for selection problems with parity-like
  tests.
\newblock Technical Report TR-392-92, Princeton University, 1993.
\newblock \url{https://www.cs.princeton.edu/research/techreps/95}.

\bibitem{KupferTies}
Ron Kupfer.
\newblock Parity tests with ties.
\newblock Preliminary version, arXiv:2604.19158v2 [cs.CC], 2026.
\newblock \url{https://arxiv.org/abs/2604.19158v2}.

\end{thebibliography}
\endgroup
\end{document}